\documentclass{article}
\usepackage[algoruled,linesnumbered,algo2e,vlined]{algorithm2e}
\usepackage{amsthm}
\usepackage{times}
\usepackage{pifont}
\usepackage{graphicx}

\usepackage{wrapfig}

\newtheorem{lemma}{Lemma}
\newtheorem{problem}{Problem}
\newtheorem{theorem}{Theorem}

\newcommand{\derive}{\mathit{val}}
\newcommand{\deriveInt}{\mathit{itv}}
\newcommand{\VarOcc}{\mathit{vOcc}}
\newcommand{\suffix}{\mathit{suf}}
\newcommand{\prefix}{\mathit{pre}}
\newcommand{\variable}[1]{X_{\langle #1\rangle}}

\newcommand{\parent}{\mathit{parent}}
\newcommand{\sstr}{\mathit{str}}
\newcommand{\depth}{\mathit{depth}}
\date{}

\title{Computing convolution on grammar-compressed text}
\author{\\
  Toshiya Tanaka, Tomohiro I, Shunsuke Inenaga,\\ Hideo Bannai, and
  Masayuki Takeda\\
  {\small Department of Informatics, Kyushu University, Fukuoka
    819-0395, Japan}\\
  {\small {\tt \{toshiya.tanaka,tomohiro.i,inenaga,bannai,takeda\}@inf.kyushu-u.ac.jp}}
}

\begin{document}
\maketitle
\begin{abstract}
The convolution between a text string $S$ of length $N$
and a pattern string $P$ of length $m$ can be computed in $O(N \log m)$ time by FFT.
It is known that various types of approximate string matching problems 
are reducible to convolution.
In this paper, we assume that the input text string is given in a compressed form,
as a \emph{straight-line program (SLP)}, which is a context free grammar
in the Chomsky normal form that derives a single string.
Given an SLP $\mathcal{S}$ of size $n$ describing a text $S$ of length $N$,
and an uncompressed pattern $P$ of length $m$,
we present a simple $O(nm \log m)$-time algorithm to compute 
the convolution between $S$ and $P$.
We then show that this can be improved to $O(\min\{nm, N-\alpha\} \log m)$ time,
where $\alpha \geq 0$ is a value that represents the amount of
redundancy that the SLP captures with respect to the length-$m$ substrings.
The key of the improvement is our new algorithm 
that computes the convolution between a trie of size $r$ and 
a pattern string $P$ of length $m$ in $O(r \log m)$ time.
\end{abstract}

\section{Introduction}

String matching is a task of find all occurrences of 
a pattern of length $m$ in a text of length $N$.
In various fields of computer science such as bioinformatics, 
image analysis and data compression,
detecting approximate occurrences of a pattern is of great importance.
Fischer and Paterson~\cite{Fischer1974SMa} found that 
various approximate string matching problems can be solved efficiently
by reduction to convolution,
and many studies have followed since.
For instance, 
it was shown in~\cite{Fischer1974SMa} that the string matching problem with don't
cares can be solved in $O(N \log m \log \sigma)$ time,
where $\sigma$ is the alphabet size.
This was later improved to $O(N \log m)$ time~\cite{Cole2002Vcm,Clifford2007Sdw}.
An $O(N \sqrt{m \log m})$-time algorithm for computing 
the Hamming distances between the pattern and all text substrings of length $m$ 
was proposed in~\cite{Abrahamson1987GSM}.
Many, if not all, large string data sets are stored in a compressed form,
and are later decompressed in order to be used and/or analyzed.
\emph{Compressed string processing (CSP)} arose from 
the recent rapid increase of digital data,
as an approach to process a given compressed string
\emph{without explicitly decompressing the entire string}.
A lot of CSP algorithms have been proposed in the last two decades,
which improve on algorithms working on uncompressed strings,
both in theory~\cite{NJC97,crochemore03:_subquad_sequen_align_algor_unres_scorin_matric,hermelin09:_unified_algor_accel_edit_distan,gawrychowski11:_LZ_comp_str_fast_} 
and in practice~\cite{shibata00:_speed_up_patter_match_text_compr,goto11:_fast_minin_slp_compr_strin,Goto2012SUq}.

The goal of this paper is
efficient computation of the convolution between
a compressed text and an uncompressed pattern.
In this paper, we assume that the input string
is represented by a \emph{straight-line program (SLP)},
which is a context free grammar in the Chomsky normal form
that derives a single string.
It is well known that outputs of various
grammar based compression algorithms~\cite{SEQUITUR,LarssonDCC99},
as well as those of dictionary compression algorithms~\cite{LZ78,LZW,LZ77,LZSS},
can be regarded as, or be quickly transformed to, SLPs~\cite{rytter03:_applic_lempel_ziv}.
Hence, algorithmic research working on SLPs is of great significance.
We present two efficient algorithms that compute 
the convolution between an SLP-compressed text of size $n$
and an uncompressed pattern of length $m$.
The first one runs in $O(nm \log m)$ time and space,
which is based on \emph{partial decompression} of 
the SLP-compressed text.
Whenever $nm = o(N)$,
this is more efficient than the existing FFT-based $O(N \log m)$-time algorithm 
for computing the convolution of a string of length $N$ and a pattern of length $m$.
However, in the worst case $n$ can be as large as $O(N)$.
Our second algorithm deals with such a case.
The key is a reduction of the covolution of an SLP and a pattern,
to the convolution of a trie and a pattern.
We show how, given a trie of size $r$ and pattern of length $m$,
we can compute the convolution between all strings of length $m$ in the
trie and the pattern in $O(r \log m)$ time.
This result gives us an $O(\min\{nm, N-\alpha\} \log m)$-time algorithm
for computing the convolution between an SLP-compressed text and a pattern,
where $\alpha \geq 0$ represents a quantity of redundancy of 
the SLP w.r.t. the substrings of length $m$.
Notice that our second method is at least as efficient as 
the existing $O(N \log m)$ algorithm,
and can be much more efficient when a given SLP is small.
Further, our result implies that \emph{any} string matching problems 
which are reducible to convolution can be efficiently solved on SLP-compressed text.

\subsection{Related work}

In~\cite{freschi10:_lz78},
an algorithm which computes the convolution between
a text and a pattern, using Lempel-Ziv 78 factorization~\cite{LZ78},
was proposed.
Given a text of length $N$ and a pattern of length $m$,
the algorithm in~\cite{freschi10:_lz78} computes the convolution 
in $O(N + mL)$ time and space,
where $L$ is the number of LZ78 factors of the text.
The authors claimed that 
$L = O(\frac{N}{\log N} h)$,
where $0 \leq h \leq 1$ is the entropy of the text.
However, this holds only on some strings over a constant alphabet,
and even on a constant alphabet 
there exist strings with $L = O(\frac{N}{\log N})$~\cite{crochemore03:_subquad_sequen_align_algor_unres_scorin_matric}.
Moreover, when the text is drawn from integer alphabet $\Sigma = [1,N]$, 
then clearly $L = \Theta(N)$.
In this case, the algorithm of~\cite{freschi10:_lz78} takes at least $O(mN)$ time 
(excluding the time cost to compute the LZ78 factorization).
Since the LZ78 encoding of a text can be seen as an SLP,
and since the running time of our algorithm is independent of the alphabet size,
this paper presents a more efficient algorithm to compute the convolution
on LZ78-compressed text over an integer alphabet.
Furthermore, our algorithm is much more general and can be applied to arbitrary SLPs.

\section{Preliminaries}

\subsection{Strings}

Let $\Sigma$ be a finite {\em alphabet}.
An element of $\Sigma^*$ is called a {\em string}.
The length of a string $S$ is denoted by $|S|$. 
The empty string $\varepsilon$ is a string of length 0,
namely, $|\varepsilon| = 0$.
For a string $S = XYZ$, $X$, $Y$ and $Z$ are called
a \emph{prefix}, \emph{substring}, and \emph{suffix} of $S$, respectively.
The $i$-th character of a string $S$ is denoted by $S[i]$, where $1 \leq i \leq |S|$.
For a string $S$ and two integers $1 \leq i \leq j \leq |S|$, 
let $S[i:j]$ denote the substring of $S$ that begins at position $i$ and ends at
position $j$.

Our model of computation is the word RAM:
We shall assume that the computer word size is at least $\log_2 |S|$, 
and hence, standard operations on
values representing lengths and positions of string $S$
can be manipulated in constant time.
Space complexities will be determined by the number of computer words (not bits).

\subsection{Convolution}
Let $V_S$ and $V_P$ be two vectors on some field whose lengths are $N$ and $m$, respectively,
with $m \leq N$.
The \emph{convolution} $C$ between $V_S$ and $V_P$ is defined by
\begin{eqnarray}
C[i] = \sum_{j=1}^{m}V_P[j] \cdot V_S[i+j-1]
\end{eqnarray}
for $1 \leq i \leq N-m+1$.
It is well-known that 
the vector $C$ can be computed in $O(N \log m)$ time by FFT.
The algorithm samples $V_S$ at 
every $(km+1)$-th position of $V_S$ for $0 \leq k \leq \lfloor \frac{N}{m} \rfloor$.
For each sampled position
the algorithm is able to compute 
the convolution between the subvector $V_S[km+1:(k+2)m]$ of length $2m$
and $V_P$ in $O(m \log m)$ time,
and therefore the whole vector $C$ can be computed in 
a total of $O(N \log m)$ time.

We can solve several types of approximate matching problems 
for a text $S$ of length $N$ and a pattern $P$ of length $m$,
by suitably mapping characters $P[j]$ and $S[i+j-1]$ to numerical values.
For example, let $\phi_{a}(x) = 1$ if $x = a$ and $0$ otherwise, for any $a \in \Sigma$,
then $\sum_{a \in \Sigma} \sum_{j=1}^{m} \phi_{a}(P[j]) \cdot \phi_{a}(S[i+j-1])$ represents 
the number of matching positions when the pattern is aligned at position $i$ of the text.
Consequently, the Hamming distances of the pattern and the text substrings 
for all positions $1 \leq i \leq N-m+1$ can be computed in a total of $O(|\Sigma| N \log m)$ time,
by computing convolution using mappings $\phi_{a}$ for all $a \in \Sigma$ and summing them up, 
which is a classic result in~\cite{Fischer1974SMa}.

For convenience, in what follows we assume strings $S$ and $P$ on integer alphabet,
and consider convolution between $S$ and $P$.

\subsection{Straight Line Programs}
\label{sec:slp}

\begin{figure}[tb]
  \centerline{\includegraphics[width=0.5\textwidth]{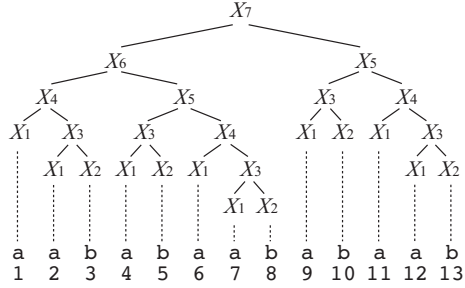}}
  \caption{
    The derivation tree of
    SLP $\mathcal S = \{ X_1 \rightarrow \mathtt{a}$, $X_2 \rightarrow \mathtt{b}$, $X_3 \rightarrow X_1X_2$,
    $X_4 \rightarrow X_1X_3$, $X_5 \rightarrow X_3X_4$, $X_6
    \rightarrow X_4X_5$, $X_7 \rightarrow X_6X_5 \}$,
    representing string $S = \derive(X_7) = \mathtt{aababaababaab}$.
  }
  \label{fig:SLP}
\end{figure}

A {\em straight line program} ({\em SLP}) is a set of assignments 
$\mathcal S = \{ X_1 \rightarrow expr_1, X_2 \rightarrow expr_2, \ldots, X_n \rightarrow expr_n\}$,
where each $X_i$ is a variable and each $expr_i$ is an expression, where
$expr_i = a$ ($a\in\Sigma$), or $expr_i = X_{\ell(i)} X_{r(i)}$~($i > \ell(i),r(i)$).
It is essentially a context free grammar in the Chomsky normal form, that derives a single string.
Let $\derive(X_i)$ represent the string derived from variable $X_i$.
To ease notation, we sometimes associate $\derive(X_i)$ with $X_i$ and
denote $|\derive(X_i)|$ as $|X_i|$,
and $\derive(X_i)([u:v])$ as $X_i([u:v])$ for any interval $[u:v]$.
An SLP $\mathcal{S}$ {\em represents} the string $S = \derive(X_n)$.
The \emph{size} of the program $\mathcal{S}$ is the number $n$ of
assignments in $\mathcal{S}$. 
Note that $|S|$ can be as large as $\Theta(2^n)$. However, we assume
as in various previous work on SLP, 
that the computer word size is at least $\log_2 |S|$, and hence,
values representing lengths and positions of $S$
in our algorithms can be manipulated in constant time.

The derivation tree of SLP $\mathcal{S}$ is a labeled
ordered binary tree where each internal node is labeled with a
non-terminal variable in $\{X_1,\ldots,X_n\}$, and each leaf is labeled with a terminal character in $\Sigma$.
The root node has label $X_n$.
Let $\mathcal{V}$ denote the set of internal nodes in
the derivation tree.
For any internal node $v\in\mathcal{V}$, 
let $\langle v\rangle$ denote the index of its label
$\variable{v}$.
Node $v$ has a single child which is a leaf labeled with $c$
when $(\variable{v} \rightarrow c) \in \mathcal{S}$ for some $c\in\Sigma$,
or
$v$ has a left-child and right-child respectively denoted $\ell(v)$ and $r(v)$,
when
$(\variable{v}\rightarrow \variable{\ell(v)}\variable{r(v)}) \in \mathcal{S}$.
Each node $v$ of the tree derives $\derive(\variable{v})$,
a substring of $S$,
whose corresponding interval $\deriveInt(v)$,
with $S(\deriveInt(v)) = \derive(\variable{v})$,
can be defined recursively as follows.
If $v$ is the root node, then $\deriveInt(v) = [1:|S|]$.
Otherwise, if $(\variable{v}\rightarrow
\variable{\ell(v)}\variable{r(v)})\in\mathcal{S}$,
then,
$\deriveInt(\ell(v)) = [b_v:b_v+|\variable{\ell(v)}|-1]$
and
$\deriveInt(r(v)) = [b_v+|\variable{\ell(v)}|:e_v]$,
where $[b_v:e_v] = \deriveInt(v)$.

For any interval $[b:e]$ of $S (1\leq b \leq e \leq |S|)$, 
let $\xi_{\mathcal{S}}(b,e)$ denote the deepest node $v$ in the derivation tree,
which derives an interval containing $[b:e]$, that is,
$\deriveInt(v)\supseteq [b:e]$,
and no proper descendant of $v$
satisfies this condition.
We say that node $v$ {\em stabs} interval $[b:e]$,
and $\variable{v}$ is called the variable that stabs the interval.
If $b = e$, we have that
$(\variable{v} \rightarrow c) \in \mathcal{S}$ for some $c\in\Sigma$,
and $\deriveInt(v) = b = e$.
If $b < e$, then
we have $(\variable{v} \rightarrow
\variable{\ell(v)}\variable{r(v)})\in\mathcal{S}$,
$b\in \deriveInt(\ell(v))$, and  $e\in\deriveInt(r(v))$.

\begin{theorem}[\cite{philip11:_random_acces_gramm_compr_strin}]
  \label{theo:random_access}
  Given an SLP $\mathcal{S} = \{ X_i \rightarrow \mathit{expr}_i \}_{i=1}^n$,
  it is possible to pre-process $\mathcal{S}$ in $O(n)$ time and
  space, so that for any interval $[b:e]$ of $S$, $1 \leq b \leq e \leq N$,
  its stabbing variable %
  $\variable{\xi_\mathcal{S}(b,e)}$
  can be computed in $O(\log N)$ time.
\end{theorem}

SLPs can be efficiently pre-processed to hold various information.
$|X_i|$ can be computed for all variables $X_i
(1\leq i\leq n)$ in a total of $O(n)$ time by a simple dynamic
programming algorithm.
Also, the following lemma is useful for partial decompression of
a prefix of a variable.

\begin{lemma}[\cite{gasieniec05:_real_time_traver_gramm_based_compr_files}]
  \label{label:prefix_decompression}
  Given an SLP $\mathcal{S} = \{ X_i \rightarrow \mathit{expr}_i \}_{i=1}^n$,
  it is possible to pre-process $\mathcal{S}$ in $O(n)$ time and
  space, so that for any variable $X_i$ and $1 \leq q \leq |X_i|$,
  the prefix of $\derive(X_i)$ of length $q$, 
  i.e. $\derive(X_i)[1:q]$, can be computed in $O(q)$ time.
\end{lemma}

\subsection{Problem}

In this paper we tackle the following problem.

\begin{problem}
Given an SLP $\mathcal{S} = \{ X_i \rightarrow \mathit{expr}_i \}_{i=1}^n$ 
describing a text $S$
and an uncompressed pattern $P$ of length $m$,
compute a compact representation of the convolution $C$
between $S$ and $P$.
\end{problem}
By ``compact representation'' above,
we mean a representation of convolution $C$ whose size 
is dependent (and polynomial) on $n$ and $m$, and not on $N = |S|$.
In the following sections, 
we will present our algorithms to solve this problem.
We will also show that given a position $i$ of the uncompressed text $S$
with $1 \leq i \leq N - m + 1$,
our representation is able to return $C[i]$ quickly.

\section{Basic algorithm}

In this section, we describe our compact representation 
of the convolution $C$ for a string $S$ represented as
an SLP $\mathcal{S}$ of size $n$ and a pattern $P$ of length $m$.
Our representation is based on the fact that the value of the 
convolution depends only on the substrings of length $m$ of $S$.
We use compact representations of all substrings of length $m$ of $S$,
which were proposed in~\cite{goto11:_fast_minin_slp_compr_strin,Goto2012SUq}.

For any variable $X_j = X_\ell X_r$, let 
$t_j = \suffix(X_\ell, m-1) \prefix(X_r,m-1)$.
Namely, $t_j$ is the substring of $\derive(X_j)$ 
obtained by concatenating the suffix of $\derive(X_\ell)$ of length at most $m-1$,
and the prefix of $\derive(X_r)$ of length at most $m-1$ (see also Figure~\ref{fig:2(m-1)}).
By the arguments of Section~\ref{sec:slp},
there exists a unique variable $X_j$ that stabs the interval $[i:i+m-1]$.
Hence, computing $C$ reduces 
to computing the convolution between $t_j$ and pattern $P$
for all variables $X_j$.

\begin{figure}[tb]
  \centerline{\includegraphics[width=0.6\textwidth]{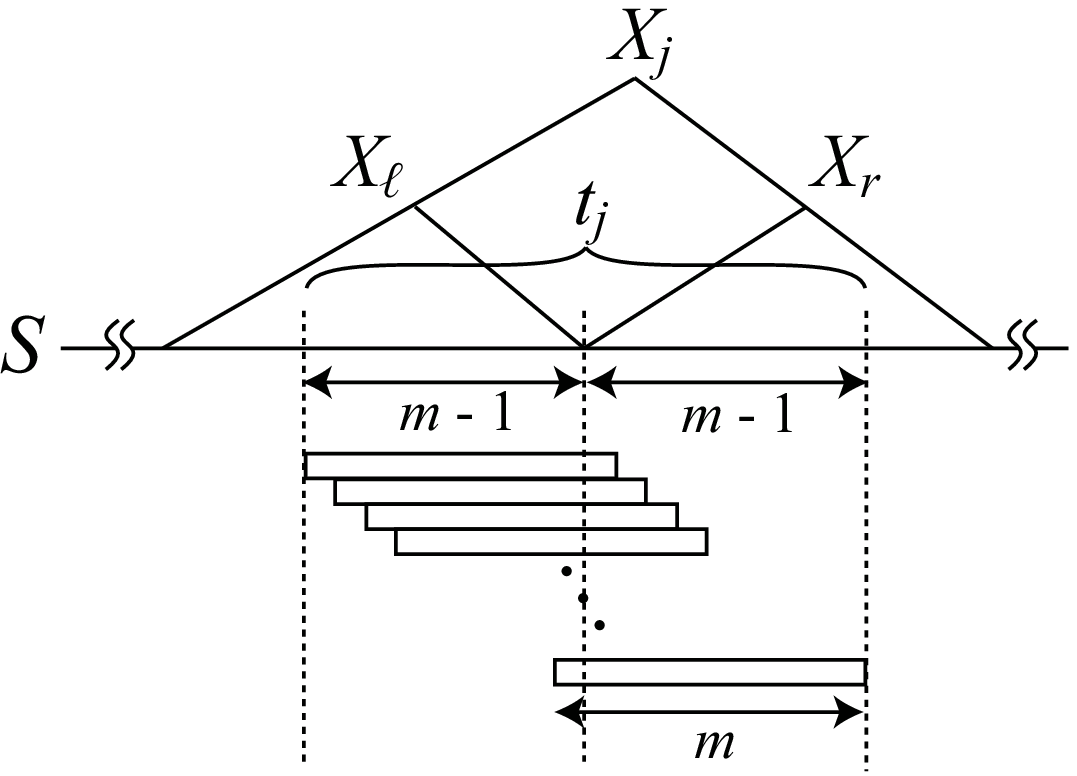}}
  \caption{Substring $t_j$ of $\derive(X_j)$.}
  \label{fig:2(m-1)}
\end{figure}

\begin{theorem}
Given an SLP $\mathcal{S}$ of size $n$ representing 
a string $S$ of length $N$, and pattern $P$ of length $m$, 
we can compute an $O(nm)$-size representation of convolution $C$ for $S$ and $P$ 
in $O(nm \log m)$ time.
Given a text position $1 \leq i \leq N - m + 1$,
our representation returns $C[i]$ in $O(\log N)$ time.
\end{theorem}

\begin{proof}
Let $t_j = \suffix(X_\ell, m-1) \prefix(X_r,m-1)$ for any variable $X_j = X_\ell X_r$.
Since $|t_j| \leq 2m-2$,
we can compute each $t_j$ in $O(m)$ time by Lemma~\ref{label:prefix_decompression}.
We then compute the convolution 
between $t_j$ and $P$ in $O(m\log m)$ time using the FFT algorithm.
Since there are $n$ variables, it takes a total of $O(n m \log m)$ time
and the total size of our representation is $O(nm)$.

By Theorem~\ref{theo:random_access} we can compute 
the stabbing variable in $O(\log N)$ time.
It is also possible to compute in $O(\log N)$ time 
the text position corresponding to the node of the derivation tree of $\mathcal{S}$
representing the stabbing variable~\cite{bannaiIT12:_LZ78_grammar_compr_}.
Thus $C[i]$ can be answered in $O(\log N)$ time.
\end{proof}

By a similar argument to Section 7 in~\cite{philip11:_random_acces_gramm_compr_strin},
we obtain the following:
\begin{theorem} \label{theo:matching}
Given a compact representation of the convolution between 
a string $S$ and a pattern $P$ described above,
we can output the set $occ$ of all approximate occurrences of $P$ in $S$
in $O(|occ|)$ time.
\end{theorem}

\section{Improved algorithm}

The algorithm of the previous section 
is efficient when the given SLP is small, i.e., $nm = o(N)$.
However, $n$ can be as large as $O(N)$,
and hence it can be slower than the existing FFT-based $O(N \log m)$-time algorithm.

To overcome this, we use the following result:
\begin{lemma}[\cite{Goto2012SUq}]
For any SLP $\mathcal{S}$ of size $n$ describing a text $S$ of length $N$,
there exists a trie $T$ of size $O(\min\{nm, N - \alpha\})$ with $\alpha \geq 0$,
such that for any substring $Q$ of length $m$ of $S$,
there exists a directed path in $T$ that spells out $Q$.
The trie $T$ can be computed in linear time in its size.
\end{lemma}

Here $\alpha$ is a value that represents the amount of
redundancy that the SLP captures with respect to the length-$m$ substrings, which is defined by
$\alpha = \sum \{(\VarOcc(X_j) - 1) \cdot (|t_j| - (m-1)) \mid |X_j| \geq m, j = 1, \dots, n \}$,
where $\VarOcc(X_j)$ denotes the number of times a variable $X_j$ occurs
in the derivation tree, i.e., 
$\VarOcc(X_j) = |\{ v \mid \variable{v}=X_j\}|$.

By the above lemma,
computing the convolution between an SLP-compressed string
and a pattern reduces to computing 
the convolution between a trie and a pattern.
In the following subsection, we will present 
our efficient algorithm to compute the convolution between a trie and a pattern.

\subsection{Convolution between trie and pattern}

Here we consider the convolution between a trie $T$ and a pattern $P \in \Sigma^{+}$.
For any node $v$ of $T$ and a positive integer $k$, 
let $\sstr_{T}(v, k)$ be the suffix of the path from the root of $T$
to $v$ of length $\min\{k, \depth(v)\}$.
The subproblem to solve is formalized as follows:
\begin{problem}\label{prob:convolution_trie}
Given a trie $T$ and a pattern $P$ of length $m$, 
for all nodes $v$ of $T$ whose depth is at least $m$,
compute $C_{T}(v) = \sum_{j = 1}^{m}\sstr_{T}(v, m)[j]P[j]$.
\end{problem}

Figure~\ref{fig:input_trie} illustrates an instance of Problem~\ref{prob:convolution_trie}.
Figure~\ref{fig:conv_52413} shows the values of the
convolution between the trie of Figure~\ref{fig:input_trie} and pattern $5 \ 2 \ 4 \ 1 \ 3$.
\begin{figure}[tb]
  \centerline{\includegraphics[width=0.6\textwidth]{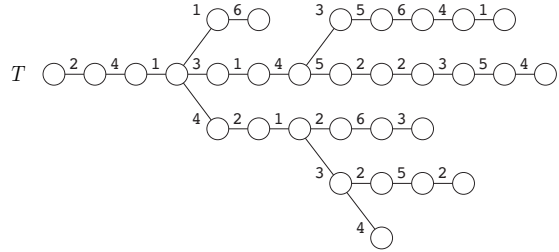}}
  \caption{Instance of input trie $T$.}
  \label{fig:input_trie}
\end{figure}
\begin{figure}[tb]
  \centerline{\includegraphics[width=0.6\textwidth]{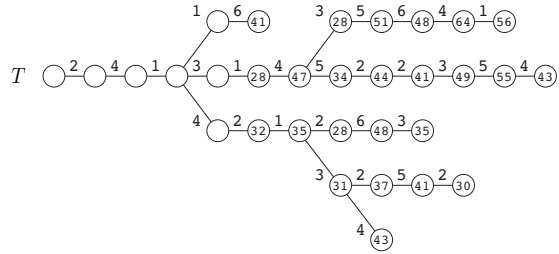}}
  \caption{The convolution between $T$ and pattern $5 \ 2 \ 4 \ 1 \ 3$. The value in each node
is the value of the convolution for the node and the pattern. The nodes with depth less than $|P| = 5$ are left blank.}
  \label{fig:conv_52413}
\end{figure}

\begin{theorem}
Problem~\ref{prob:convolution_trie} can be solved in $O(r \log m)$ time, 
where $r$ is the size of $T$ and $m = |P|$.
\end{theorem}

\begin{proof}
Assume that the height of $T$ is at least $m$ since otherwise no computation is needed.
We show how to compute $C_{T}(v)$ in $O(\log m)$ amortized time for each node $v \in T$.
We consider the \emph{long path decomposition}
such that $T$ is decomposed into its longest path 
and a forest consisting of the nodes that are not contained in the longest path.
We recursively apply the above decomposition to all trees in the forest,
until each subtree consists only of a single path.
Figure~\ref{fig:trie_lpd} shows the long path decomposition 
of the trie shown in Figure~\ref{fig:input_trie}.
\begin{figure}[tb]
  \centerline{\includegraphics[width=0.6\textwidth]{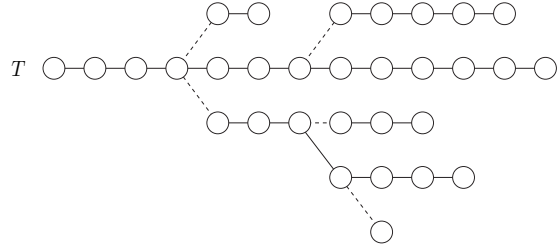}}
  \caption{The long path decomposition of the trie shown in Figure~\ref{fig:input_trie}.}
  \label{fig:trie_lpd}
\end{figure}
It is easy to see that we can compute the long path decomposition in $O(r)$ time.
For each path, we compute the convolution by FFT.
Let $(w_{1}, w_{2}, \dots, w_{d})$ be one of the long paths, 
where $d$ is the number of nodes on the path.
\begin{itemize}
  \item When $d \geq m$: 
        It is enough to compute the convolution between $\sstr_{T}(w_{d}, d+m-1)$ and $P$, 
        which takes $O((d+m-1) \log m)$ time, i.e., $O(\frac{(d+m-1)}{d} \log m) = O(\log m)$ time per node.
      \item When $d < m$: The same method costs too much, i.e.,
        $O(\frac{(d+m-1)}{d} \log m)$ time per node, and thus we need
        a trick.
        The assumption that the height of $T$ is at least $m$ implies that $w_{1}$ is not the root of $T$ 
        since otherwise, the longest path in $T$ would be
        $(w_{1}, w_{2}, \dots, w_{d})$,
        and $d-1 (< m)$ would be the height of $T$, a contradiction.
        Consequently, from the definition of the long path decomposition, 
        there must exist a path (not necessarily a long path) $(z_{1}, z_{2}, \dots, z_{d})$ such that $w_{1} \neq z_{1}$ and $\parent(w_{1}) = \parent(z_{1})$.
        For any $1 \leq i \leq d$ with $\depth(w_{i}) \geq m$, $C_{T}(w_{i})$ can be written as follows:
        \begin{eqnarray*}
            C_{T}(w_{i}) &=& \sum_{j = 1}^{m-d}\sstr_{T}(w_{i}, m)[j]P[j] + \sum_{j=m-d+1}^{m}\sstr_{T}(w_{i}, m)[j]P[j]\\
            &=& \sum_{j = 1}^{m-d}\sstr_{T}(z_{i}, m)[j]P[j] + \sum_{j=m-d+1}^{m}\sstr_{T}(w_{i}, m)[j]P[j]\\
            &=& C_{T}(z_{i}) - \sum_{j=m-d+1}^{m}\sstr_{T}(z_{i}, m)[j]P[j] + \sum_{j=m-d+1}^{m}\sstr_{T}(w_{i}, m)[j]P[j]\\
            &=& C_{T}(z_{i}) - C'_{T}(z_{i}) + C'_{T}(w_{i}),
        \end{eqnarray*}
        where $C'_{T}(v) = \sum_{j=m-d+1}^{m}\sstr_{T}(v, m)[j]P[j]$.
        For all $1 \leq i \leq d$, $C'_{T}(w_{i})$ (resp. $C'_{T}(z_{i})$) can be computed in $O((d+d-1) \log d)$ time 
        by convolution between $\sstr_{T}(w_{d}, d+d-1)$ (resp. $\sstr_{T}(z_{d}, d+d-1)$) and $P[m-d+1:m]$.
        Therefore, assuming that $C_{T}(z_{i})$ is already computed for all $1 \leq i \leq d$, 
        we can compute $C_{T}(w_{i})$ for all $1 \leq i \leq d$ in $O(\frac{(d+d-1)}{d} \log d) = O(\log m)$ time per node.
\end{itemize}
It follows from the above discussion that we can solve Problem~\ref{prob:convolution_trie} in $O(r \log m)$ time
by computing values of convolution by the longest path first and making use of the result when encountering a short path whose length is less than $m$.
\end{proof}

We obtain the main result of this paper:
\begin{theorem}
Given SLP $\mathcal{S}$ of size $n$ representing 
a string $S$ of length $N$, and pattern $P$ of length $m$, 
we can compute an $O(\min\{nm, N-\alpha\})$-size representation of convolution $C$ for $S$ and $P$ 
in $O(\min\{nm, N-\alpha\} \log m)$ time, where $\alpha \geq 0$.
Given a text position $1 \leq i \leq N - m + 1$,
our representation returns $C[i]$ in $O(\log N)$ time.
\end{theorem}

We note that a similar result to Theorem~\ref{theo:matching}
holds for our $O(\min\{nm, N-\alpha\})$-size representation of convolution,
and hence we can compute all approximate occurrences in time linear in its size.

\section{Conclusions and future work}

In this paper we showed how, 
given an SLP-compressed text of size $n$ and an uncompressed pattern of length $m$, 
we can compute the convolution between the text and the pattern efficiently.
We employed an $O(\min\{nm, N-\alpha\})$-size trie representation of all substrings 
of length $m$ in the text, 
which never exceeds the uncompressed size $N$ of the text.
By introducing a new technique to compute the convolution between a trie of size $r$ 
and a pattern of length $m$ in $O(r \log m)$ time, 
we achieve an $O(\min\{nm, N-\alpha\} \log m)$-time solution to the problem.
A consequence of this result is that, 
for any string matching problem reducible to convolution,
there exists a CPS algorithm that does not require decompression
of the entire compressed string.

However, it is not yet obvious whether 
we can straightforwardly adapt an algorithm which also uses 
techniques other than convolution, such as the one in~\cite{Amir2004Fas}.
Future work of interest is to clarify the above matter,
and to implement our algorithms and conduct experiments on highly compressible texts.

\bibliographystyle{plain}
\bibliography{ref}

\end{document}